\documentclass[10pt]{article}
\usepackage{amssymb, amsmath, amsthm, bm, float, graphicx,  tabularx, geometry, mathrsfs, setspace, enumerate, appendix, csquotes, caption, url, yfonts, pslatex, sectsty, eulervm, xcolor, hyperref}
\usepackage[T1]{fontenc}
\usepackage[round]{natbib}

\newtheorem*{definition}{Definition}

\newtheorem{proposition}{Proposition}

\newtheorem{remark}{Remark}

\newcommand{\dom}{\text{dom}}

\newcommand{\N}{\mathbb{N}}

\newcommand{\ran}{\text{ran}}

\newcommand{\R}{\mathbb{R}}

\newcommand{\e}{\varepsilon}

\input{tcilatex}

\title{Extended Gini index}
\author{Ram Sewak Dubey%
\thanks{Department of Economics, Feliciano School of Business, Montclair State
University, Montclair, NJ 07043; E-mail: dubeyr@montclair.edu} \and %
Giorgio Laguzzi%
\thanks{University of Freiburg in the Mathematical Logic Group at Eckerstr. 1, 79104
Freiburg im Breisgau, Germany; Email: giorgio.laguzzi.1984@gmail.com}}

\date{\today}
\begin{document}
\maketitle

\begin{abstract}
We propose an extended version of Gini index defined on the set of infinite utility streams, $X=Y^\N$ where $Y\subset \R$.
For $Y$ containing at most finitely many elements, the index satisfies the generalized Pigou-Dalton transfer principles in addition to the anonymity axiom.

\noindent \emph{Keywords:} \texttt{Anonymity,}\; \texttt{Extended Gini index,}\; \texttt{Generalized Pigou-Dalton transfer principle,}\;  \texttt{Social welfare function.}

\noindent \emph{Journal of Economic Literature} Classification Numbers: \texttt{C65,} \texttt{D63,}\; \texttt{D71.}
\end{abstract}

\newpage

\section{Introduction}\label{s1}

The Gini index (also referred to as Gini Coefficient, \citet{gini1997} is a widely used measure of inequality in income or wealth distribution in the society.
It is a measure of dispersion of income distribution for a given population and is sensitive to redistribution of income from rich to poor.
In this paper, we propose an extended version of the Gini index as a real valued representation of the infinite utility streams.
The new index satisfies a generalized version of Pigou-Dalton transfer principle and the anonymity axiom.

A brief review of the two equity axioms is as follows.
Anonymity axiom is an example of \emph{procedural equity}.
It applies to the situations where the changes involved in the infinite utility streams do not alter the distribution of utilities.%
\footnote{The idea of anonymity was introduced in a classic contribution, \citet{ramsey1928}, who observed that discounting one generation's utility relative to another's is \enquote{ethically indefensible}, and something that \enquote{arises merely from the weakness of the imagination}. 
\citet{diamond1965} formalized the concept of \enquote{equal treatment} of all generations (present and future) in the form of an \emph{anonymity} axiom on social preferences.} 
The anonymity axiom requires that the society should be indifferent between two streams of well-being, if one is obtained from the other by interchanging the well-being of any pair of generations.

The second equity concept is an example of the \emph{consequentialist equity}.
The particular version considered in this paper is the well-known Pigou-Dalton transfer principle.%
\footnote{The inequality reducing property was initially hinted at by \citet[p. 24]{pigou1912} as \enquote{The Principle of Transfers}.
\citet[p. 351]{dalton1920} described it as \enquote{If there are only two income-receivers and a transfer of income takes place from the richer to the poorer, inequality is diminished.}.
The version of equity axiom described here for the infinite utility streams was introduced and discussed in \citet{sakai2006}, \citet{bossert2007}, and \citet{hara2008}.}
Pigou-Dalton transfer principle compares two infinite utility streams ($x$ and $y$) in which all generations except two have the same utility levels in both utility streams; regarding the two remaining generations (say, $i$ and $j$), if  $y_i < x_i < x_j < y_j$ and $y_i+y_j=x_i+x_j$, then utility stream $x$ is socially preferred to $y$. 
This equity principle ranks utility sequence $x$ superior to $y$ as $x$ is obtained from $y$ by carrying out a non-leaky and non-rank-switching transfer of welfare from a rich to a poor generation.

It is easy to infer that this definition would also help us in ranking utility sequence $x$ superior to $y$ if $x$ is obtained from $y$ by carrying out  non-leaky and non-rank-switching transfers of welfare among any finitely many pairs of generations.
To enable us to compare sequence $x$ and $y$ when $x$ is obtained from $y$ by carrying out arbitrarily many (possibly infinitely many) pairs of rich and poor generations, we have introduced a generalized version of the Pigou-Dalton transfer principle in \citet{dubey2020a}.
It has been shown that the generalized infinite Pigou-Dalton transfer principle admits real-valued representation if and only if $Y$ does not contain more than seven distinct elements.

In this short note we consider a weaker version of the generalized Pigou-Dalton principle and show (in Proposition \ref{P1}) that an index defined along the lines of the Gini index, which we call the \emph{extended Gini index}, satisfies the Pigou-Dalton principle in addition to the anonymity axiom when $Y$ contains finitely many distinct elements.

The rest of the paper is organized as follows.
Section \ref{s2} contains the relevant definitions and in section \ref{s3} we discuss the generalized Pigou-Dalton transfer principle which is the focus of our paper.
Section \ref{s4} contains the two results with detailed proofs.
We conclude in section \ref{s5}.

\section{Preliminaries}\label{s2}
\subsection{Notations and definitions}\label{s21}
Let $\R$ and $\N$ be the sets of real numbers and natural numbers respectively. 
For all $y, z\,\in \R^{\N}\,$, we write $y\geq z$ if $y_{n}\geq z_{n}$, for all $n\in \N$; we write $y>z$ if $y\geq z$ and $y\neq z$; and we write $y\gg z$ if $y_{n}>z_{n}$ for all $n\in \N$.

A \emph{partial function} $f: X \rightarrow Y$ is a function from a subset $S$ of $X$ to $Y$. 
If $S$ equals $X$, the partial function is said to be total.
Domain and range of function $f$ are denoted by $\dom(f)$ and $\ran(f)$ respectively.
\begin{definition}
\emph{A partial function $\alpha: \N \rightarrow \N$ is called a \emph{pairing function} if and only if $\alpha$ satisfies $\forall n \in \dom(\alpha)$, $\alpha(\alpha(n))=n$.}
\end{definition}
Note that $\dom(\alpha)=\ran(\alpha)$ for every pairing function $\alpha$. 
We denote the set of all pairing functions by $\Pi$.

\subsubsection{Density for subsets of natural numbers}
Let  $|S|$ denote the cardinality of the finite set $S\subset \N$.
The \emph{lower asymptotic density} of $S\subset \N$ is defined as:
\begin{equation*}
\underline{d} (S) = \underset{n\rightarrow \infty}{\liminf}\; \frac{|S\cap\{1, 2, \cdots, n\}|}{n}.
\end{equation*}
Likewise, the \emph{upper asymptotic density} of $S\subset \N$ is defined as:
\begin{equation*}
\overline{d} (S) = \underset{n\rightarrow \infty}{\limsup}\; \frac{|S\cap\{1, 2, \cdots, n\}|}{n}.
\end{equation*}
If the two coincide for a set $S\subset\N$, it is called the \emph{asymptotic density} of set $S$, $d(S)$.

\subsubsection{Social Welfare Order}
Let $Y$, a non-empty subset of $\R$, be the set of all possible utilities that any generation can achieve. 
Then $X\equiv Y^{\N}$ is the set of all possible utility streams. 
If $\langle x_{n}\rangle \;\in \;X$, then $\langle x_{n}\rangle = (x_{1}, x_{2}, \cdots)$, where, for all $n\in \N$, $x_{n}\in Y$ represents the amount of utility that the generation of period $n$ earns.
We consider transitive binary relations on $X$, denoted by $\succsim$ and called \emph{social welfare relation} (SWR), with symmetric and asymmetric parts denoted by $\sim$ and $\succ$ respectively, defined in the usual way. 
A \emph{social welfare order} (SWO) is a complete social welfare relation.
A \emph{social welfare function} (SWF) is a mapping $W:X\rightarrow \R$. 
Given a SWO $\succsim$ on $X$, we say that $\succsim$ can be \emph{represented} by a real-valued function if there is a mapping $W:X\rightarrow \R$ such that for all $x, y\in X$, we have $x\succsim y$ if and only if $W(x)\geq W(y)$.

\subsubsection{Equity axioms}
The following axioms on social welfare orders are used in the analysis. 
\begin{definition}
\emph{(Anonymity - AN):\quad 
If $x,y\;\in \;X$, and if there exist $i,j\;\in \;\N$ such that $x_{i} = y_{j}$ and $x_{j} = y_{i}$, and for every $k\in \N\setminus \{i,j\}$, $x_{k} = y_{k}$, then $x\sim y$.}
\end{definition}

\begin{definition}
\emph{(Pigou-Dalton transfer principle - PD): \quad 
If $x,y\in X$, and there exist $i,j\in \N$ and $\varepsilon>0$, such that $x_{i}+ \e = y_{i} < y_{j} = x_{j}-\e$, while $y_{k} = x_{k}$ for all $k\in \N\setminus \{i,j\}$, then $y\succ x$.}
\end{definition}
Anonymity is an example of procedural equity.
Pigou-Dalton transfer principle is an example of consequentialist equity.
\section{Generalized Pigou-Dalton transfer principles}\label{s3}

The key observation and reason of the present paper for studying some extended versions of these consequentialist equity principles is motivated by the following observation.
Given a set of utilities $Y:= \{1, 2, 3, 4, 5\}\subset\R$, consider the two infinite streams  
\[
x:= \langle 2, 3, 5, 2, 3, 5,  2, 3, \cdots \rangle, \;\text{and}\; y:= \langle 1, 4, 5, 1, 4, 5, 1, 4, \cdots \rangle.
\] 
Following the expected interpretation of a redistributive equity principle, we should be able to always rank $y \prec x$. 
In the finite case, PD together with transitivity is sufficient to secure such a ranking, but in the infinite case transitivity cannot be extended to infinite chains. 
Hence PD even with transitivity is not a sufficient condition to secure the desired ranking $y \prec x$, and so an extension is necessary.

\begin{definition}
\emph{(Generalized Pigou-Dalton, GPD):\quad 
Given $x, y \in X$ if there is $\alpha \in \Pi$ such that for every $j \notin \dom(\alpha), x_j=y_j$, and for every $i \in \dom(\alpha)$ there is $\varepsilon_i >0$ such that
\[
\text{either}\; y_i + \varepsilon_i = x_i < x_{\alpha(i)} = y_{\alpha(i)} - \varepsilon_i \; \text{or}\; y_{\alpha(i)} +\varepsilon_i < x_{\alpha(i)} < x_i < y_i - \varepsilon_i,\;\text{then}\; y \prec x.
\]}
\end{definition}

In \citet{dubey2020a}, we have investigated the existence and representation of these generalized equity principles. 
The results show that when we do not put any further restriction to the combinatorial characteristics of the pairing function, representation of SWRs satisfying those principles is rather demanding and hard to obtain. 
This leads us to investigate more restricted and weaker forms of  generalized Pigou-Dalton transfer.    
\underline{A first line} of weaker variants of GPD is given by imposing some combinatorial restrictions on the pairing functions as explained in what follows. 
Consider two streams $x, y \in Y^\N$, with $Y=\{1, 2, 3, 4\}$:
\[
x(n)= \left\{ 
\begin{array}{ll}
1 & \exists k \in \N (n=10^k) \\
4 & \text{else} \\
\end{array}
\right.
\quad
y(n)= \left\{ 
\begin{array}{ll}
2 & \text{if}\; x(n) = 1 \\
3 & \text{else.} \\
\end{array}
\right.
\] 
By GPD we get $x \prec y$ even if the welfare improving re-distributions, from $(1, 4)$ in $x$ to $(2, 3)$ in $y$ occurs via a pairing function $\alpha$ such that for every $N \in \N$ there exists $n \in \dom(\alpha)$ such that $|n - \alpha(n)|> N$. 
In other words, the distances between the generations linked by the pairing function $\alpha$ grows in an unbounded manner. 
To avoid this feature, we can require the pairing function $\alpha$ having some limitations in term of being bounded, in a similar fashion as anonymity could require some restrictions on the family of permutations. In this paper we focus on the following types of pairing functions.
\begin{definition} \label{Def1} 
\emph{We say that $\alpha \in \Pi$ is a \emph{fixed-step pairing function} ($\alpha \in \Pi^s$) if and only if there exists $h \in \N$ (called the \emph{step of $\alpha$}) such that $\forall n \in \N \forall k \in ((n-1) h, n h]$, one has $\alpha(k) \in ((n-1) h, n h]$.}
\end{definition}
One can then introduce the following weakening of GPD.

\begin{definition}
\emph{(Fixed-step Generalized Pigou-Dalton, s-GPD):\quad 
There exists $h \in \N$ such that for every $x, y \in X$ if there is $\alpha \in \Pi^s$ with step $h$ such that for every $j \notin \dom(\alpha), x_j=y_j$, and for every $i \in \dom(\alpha)$ there exists $\e_i>0$ such that:}
\[
\emph{ either } y_i + \e_i= x_i < x_{\alpha(i)} = y_{\alpha(i)} -\e_i \emph{ or } y_{\alpha(i)} +\e_i = x_{\alpha(i)} < x_i = y_i - \e_i,
\]
\emph{then $y \prec x$.}
\end{definition}
Note that GPD $\Rightarrow$ s-GPD. 
\underline{A second line} works as follows. 
Since we are considering an infinite time horizon, the ranking induced by $\prec$ should be sensitive not only to few changes, but to a number of changes as large as possible. 
For instance, in an infinite setting one could require that the number of individuals/generations linked via the pairing function (where we can appreciate a reduction of inequality) should be at least an infinite set, and possibly with some non-zero density. 
This is in line with the weaker forms of Pareto principles extensively studied in the literature, such as infinite Pareto, asymptotic Pareto and weak Pareto. 
The following definitions capture this relevant idea for our study.

\begin{definition}
\begin{itemize}
\item \emph{(Infinite Pigou-Dalton, IPD):\quad 
Given $x, y \in X$ if there is $\alpha \in \Pi$ such that $\dom(\alpha)$ is infinite, for every $j \notin \dom(\alpha), x_j=y_j$, and for every $i \in \dom(\alpha)$ one has
\[
\text{ either } y_i + \e_i= x_i < x_{\alpha(i)} = y_{\alpha(i)} -\e_i \text{ or } y_{\alpha(i)} +\e_i = x_{\alpha(i)} < x_i = y_i - \e_i,\,\text{then}\; y \prec x.
\]}
\item \emph{(Asymptotic Pigou-Dalton, APD):\quad
Given $x, y \in X$ if there is $\alpha \in \Pi$ such that $d(\dom(\alpha))>0$, for every $j \notin \dom(\alpha), x_j=y_j$, and for every $i \in \dom(\alpha)$ one has
\[
\text{ either } y_i + \e_i= x_i < x_{\alpha(i)} = y_{\alpha(i)} -\e_i \text{ or } y_{\alpha(i)} +\e_i = x_{\alpha(i)} < x_i = y_i - \e_i,\,\text{then}\; y \prec x.
\]}
\item \emph{(Weak Pigou-Dalton, WPD):\quad
Given $x, y \in X$ if there is $\alpha \in \Pi$ such that $\dom(\alpha)=\N$, for every $j \notin \dom(\alpha), x_j=y_j$, and for every $i \in \dom(\alpha)$ one has
\[
\text{ either } y_i + \e_i= x_i < x_{\alpha(i)} = y_{\alpha(i)} -\e_i \text{ or } y_{\alpha(i)} +\e_i = x_{\alpha(i)} < x_i = y_i - \e_i,\,\text{then}\; y \prec x.
\]}
\end{itemize}
\end{definition}
Note that GPD $\Rightarrow$ IPD $\Rightarrow$ APD $\Rightarrow$ WPD.
In \citet{dubey2020a} we focus on GPD, IPD and WPD, whereas in this paper we focus on versions of APD. 

\begin{remark}
\emph{Note that we can then easily combine the two types of weaker variants, and obtain principles like s-APD, where we require both that the pairing function $\alpha \in \Pi^s$ and that $d(\dom(\alpha))>0$.}
\end{remark}

In \citet[Lemma 2]{dubey2020a}, we have proven that if the utility domain rules out certain types of order-structure one can always prove the existence of SWR satisfying GPD (also combined with AN and M). 
Since all the variants we are considering in this context are weaker versions of GPD, those results suffice to obtain the existence of corresponding SWRs. 
In the next section we study more deeply exactly the various cases and understand when combinations of these generalized equity principles with AN are representable and when, on the contrary, they lead to impossibility of representation by social welfare functions. We present a positive and a negative result, and discuss further developments we would tackle in future research.

\section{Representation of fixed-step asymptotic equitable social welfare relations}\label{s4}
In \citet{dubey2020a} we have proven that any SWO satisfying IPD is not representable when the utility domain $Y$ has at least eight elements. 
A careful scrutiny of the proof shows that actually non-representability persists even if we weaken IPD to APD, since the pairing functions defined in that proof actually has domain with strictly positive density. 
But what is crucial regarding those pairing functions is that they do not satisfy any particular characteristics in line with Definition \ref{Def1}. 
The following result shows that if we put some restrictions on the structure of the pairing functions, namely fixed-step with asymptotic density $>0$, then we obtain an elegant social welfare function, which recalls an extended infinite version of the well-known Gini index.

\begin{proposition}\label{P1}
Let $X= Y^{\N}$ where $|Y|<\infty$ and $Y \subseteq [0,1]$. 
Then there exists a social welfare function $W: X \rightarrow \R$ satisfying fixed step asymptotic Pigou-Dalton (s-APD) and anonymity (AN) axioms.
\end{proposition}

\begin{proof}
We present the result for $Y:=\{1, 2, \cdots, M\}$, i.e., we get $W: Y^{\N} \rightarrow \R$ satisfying s-APD and AN.   
Let $h \in \N$ be the fixed step, $I_n:= ((n-1)h, nh]$ for $n \in \N$ and $H_n := \sup I_n$. 
Define
\begin{equation}\label{P8E1}
W_N(x):=  \frac{1}{H_N^2} \sum_{k=1}^{H_N} \sum_{j=1}^{H_N} |x(k) - x(j)|,
\end{equation}
and 
\begin{equation}\label{P8E2}
W(x):= -\liminf_{N\rightarrow \infty} \; W_N(x),
\end{equation}
We claim $W$ satisfies AN and s-APD.  
Anonymity is  trivial to show.
We need to prove $W$ satisfies s-APD.
Let $x, y \in Y^\N$ be such that there exists $\alpha \in \Pi^s$  with $d(\dom(\alpha))>0$ and 
\begin{itemize}
\item for all $k \in \dom(\alpha)$, either $x_k < y_k < y_{\alpha(k)} < x_{\alpha(k)}$ or $x_{\alpha(k)} < y_{\alpha(k)} < y_k < x_k$; 
\item for all $k \notin \dom(\alpha)$, $x_k=y_k$.
\end{itemize}
In order to show $W(x) < W(y)$, we first need to compare 
\[
\sum_{k=1}^{H_N} \sum_{j=1}^{H_N} |x_k - x_j| \quad \text{and} \quad \sum_{k=1}^{H_N} \sum_{j=1}^{H_N} |y_k - y_j|.
\]
The choice of $x$ and $y$ reveals that, for every $n \in \N$,
\[
\begin{split}
\sum_{k=1}^{H_n} \sum_{j=1}^{H_n} |x(k) - x(j)| \geq \sum_{k=1}^{H_n} \sum_{j=1}^{H_n} |y(k)-y(j)| + \frac{2}{5} |\dom(\alpha) \cap [1,H_n]| \cdot  |\dom(\alpha) \cap [1,H_n]|. \\
\end{split}
\] 
To show that, we can proceed by an inductive argument on all pairs $(k,j) \in H_n \times H_n$, by computing the values $|x_k - x_j|$'s compared to $|y_k-y_j|$'s. 
For every $j, k \in H_N$ we have four possible non-trivial cases: 
\begin{enumerate}[1)]
\item{$k = \alpha(j)$ or $j=\alpha(k)$:  
Then $\e_j=\e_k$ and trivially $|x_j-x_k| = |y_j-y_k| + 2\e_j$.}
\item{$x_k = y_k$ and $x_j = y_j$: 
Then $|x_k - x_j|=|y_k-y_j|$.}
\item{$x_j = y_j$ and $y_k < x_k$:
For $\{j,  k \}$ and $\{ j,\alpha(k)\}$ and note that $|x_j - x_k|+ |x_j-x_{\alpha(k)}| = |y_j-y_k| + |y_j-y_{\alpha(k)}| + \varepsilon_k - \e_k= |y_j-y_k| + |y_j-y_{\alpha(k)}|$.}
\item{$y_k < x_k$ and $x_j < y_j$:
We compute the values given by $\{k, j\}$, $\{k,\alpha(j)\}$, $\{j, \alpha(k)\}$, $\{j, \alpha(j)\}$, $\{\alpha(j), \alpha(k)\}$.
Combinatorial computations provide the following:
\begin{itemize}
\item $|x_k - x_j| \geq |y_k - y_j| + \mu_0 \e_k + \mu_0 \e_j$, where $\mu_0$ takes a values in $\{-1,1\}$ depending on $x_k, y_k, x_j, y_j$ using the following criteria: $\mu_0=1$ if $y_k \geq y_j$, $\mu_0=-1$ if $y_j \geq y_k$.  
\item $|x_k - x_{\alpha(j)}| \geq |y_k - y_{\alpha(j)}| + \mu_1 \e_k - \mu_1 \e_j$, where $\mu_1$ takes a values in $\{-1,1\}$ depending on $x_k, y_k, x_{\alpha(j)}, y_{\alpha(j)}$ using the following criteria: $\mu_1=1$ if $x_{\alpha(j)} < y_k$, $\mu_1=-1$ if $x_k < y_{\alpha(j)}$, and $\mu_1 \in \{-1,1\}$ is similarly chosen when $y_{\alpha(j)} \leq x_k$ or $y_k \leq x_{\alpha(j)}$, and depending on whether $\e_k \leq \e_j$ or not. 
\item $|x_j - x_{\alpha(k)}| \geq |y_j - y_{\alpha(k)}| - \mu_2 \e_k + \mu_2 \e_j$, where $\mu_2$ takes a values in $\{-1,1\}$ depending on $x_j, y_j, x_{\alpha(k)}, y_{\alpha(k)}$ using the following criteria: $\mu_2=1$ if $x_{\alpha(k)} > y_j$, $\mu_2=-1$ if $x_j > y_{\alpha(k)}$, and $\mu_2 \in \{-1,1 \}$ is similarly chosen when $y_{\alpha(k)} \geq x_j$ or $y_j \geq x_{\alpha(k)}$ and $\e_k \leq \e_j$, and depending on whether $\e_k \leq \e_j$ or not.
\item $|x_{\alpha(k)} - x_{\alpha(j)}| \geq |y_{\alpha(k)} - y_{\alpha(j)}| - \mu_3 \e_k - \mu_3 \e_j$, where $\mu_3$ takes a values in $\{-1,1\}$ depending on $x_{\alpha(k)}$, $y_{\alpha(k)}$, $x_{\alpha(j)}$, $y_{\alpha(j)}$ using the following criteria: 
$\mu_3=1$ if $x_{\alpha(k)} > x_{\alpha(j)}$, $\mu_3=-1$ if $y_{\alpha(k)} < y_{\alpha(j)}$, and $\mu_3 \in \{-1,1\}$ is similarly chosen when $y_{\alpha(k)} \geq y_{\alpha(j)}$ or $x_{\alpha(k)} \leq x_{\alpha(j)}$ and $\e_k \leq \e_j$, and depending on whether $\e_k \leq \e_j$ or not.
\item $|x_j - x_{\alpha(j)}| = |y_j - y_{\alpha(j)}| + 2 \e_j$.
\end{itemize}
All together we obtain: 
\begin{equation}  \label{eq1} 
\begin{split}
& |x_k - x_j| + |x_k - x_{\alpha(j)}|+ |x_j - x_{\alpha(k)}| + |x_{\alpha(k)} - x_{\alpha(j)}| + |x_j - x_{\alpha(j)}| \geq \\
& \geq |y_k - y_j| + |y_k - y_{\alpha(j)}|+ |y_j - y_{\alpha(k)}| + |y_{\alpha(k)} - y_{\alpha(j)}| + |y_j - y_{\alpha(j)}| + \\
&\e_k (\mu_0 + \mu_1 - \mu_2 - \mu_3) + \e_j (\mu_0 - \mu_1 + \mu_2 - \mu_3) + 2 \e_j.
\end{split}
\end{equation}
By construction we have $\mu_1 \leq \mu_0$, $\mu_3 \leq \mu_2$, $\mu_2 \leq \mu_0$ and $\mu_3 \leq \mu_1$. 
Hence the last line is $\geq 2 \e_j$.}
\end{enumerate}

The other combinations, like $j=k$ or $y_k < x_k$ and $y_j < x_j$, are simply reducible to one of these four cases, or analogous proof-arguments.
Proceeding inductively and comparing all pairs through the induction on $j$ and $k$ following the four cases, we can therefore observe two facts.
Firstly,  cases 2) and 3) show that whenever at least one of the pairs involved does not belong to $\dom(\alpha)$, the sum of the absolute values of symmetric differences of the considered combinations of stream $x$ and stream $y$ coincide.
Secondly, case 1) and 4) show that whenever the pairs both belong to $\dom(\alpha)$, then the combinations considered always reveal that the sum involving the values of stream $x$ is strictly larger than the ones referring to $y$. 
Hence, putting these two observations together, we obtain:   
\[
\sum_{k=1}^{H_N} \sum_{j=1}^{H_N} |x_k-x_j| \geq \sum_{k=1}^{H_N} \sum_{j=1}^{H_N} |y_k-y_j| + \frac{2}{5} \e \cdot |\dom(\alpha) \cap [1,H_N]| \cdot |\dom(\alpha) \cap [1,H_N]|,
\]
where the fraction $\frac{2}{5}$ comes from the counting in (\ref{eq1}), where it is shown that for a pair $j, k \in \dom(\alpha)$ the sum of the 5 combinations considered for $x$ overcome the analog sum for $y$ by a factor $2\e:=2\min\left\{\e_j, \e_k\right\}$, which gives that the total double sum for $x$ is larger than the total double sum for $y$ by $2\e$ over $5$. 
Therefore we obtain the desired property, as by the characteristics of $Y$ it holds $\e \geq 1$.
Hence we get:
\[
\begin{split}
-W(x)&= \liminf_{N \rightarrow \infty} \frac{1}{H_N^2} \sum_{k=1}^{H_N} \sum_{j=1}^{H_N} |x(k) - x(j)| \\
&\geq \liminf_{N \rightarrow \infty} \frac{1}{H_N^2}  \Big ( \sum_{k=1}^{H_N} \sum_{j=1}^{H_N} |y(k)-y(j)|+ \frac{2}{5} |\dom(\alpha) \cap [1,H_N]| \cdot |\dom(\alpha) \cap [1,H_N]| \Big )\\
&\geq \liminf_{N \rightarrow \infty} \frac{1}{H_N^2}  \Big ( \sum_{k=1}^{H_N} \sum_{j=1}^{H_N} |y(k)-y(j)| \Big)+ \frac{2}{5} \liminf_{N \rightarrow \infty} \frac{1}{H_N^2}  \Big (|\dom(\alpha) \cap [1,H_N]| \Big )^2\\
&\geq -W(y) + \frac{2}{5}\underline{d}^2(\dom(\alpha)).
\end{split}
\]
Since by assumption $\underline{d}(\dom(\alpha))>0$, we therefore get $W(x) < W(y)$ as desired. 
\end{proof}

Note that in the inequalities we have used the property $\liminf(a+b) \geq \liminf \;a + \liminf\; b$, $\liminf(a\cdot b) \geq \liminf \;a \cdot \liminf\; b$ and $\dom(\alpha)$ is strictly positive.
Therefore the proof cannot be adopted to work for s-IPD as well, and this is perfectly coherent with [\citet[Theorem 1]{dubey2020a}] where it is shown that the combination of s-IPD and AN is not representable for every non-trivial utility domain.

We conclude by stating an impossibility result, whose proof is delegated to a future work. 
Proposition \ref{P2} below shows that Proposition \ref{P1} cannot be extended if the utility domain becomes too complicated.
Specifically, Proposition \ref{P2} shows a limitation when $Y \subseteq [0,1]$ contains some infinite subsets with particular order type. 
More specifically, if the set $Y$ contains a pair of infinite sets one increasing and the other decreasing with well-defined minimum or maximum elements for each subset of $Y$, then s-APD and AN together are not representable.

\begin{proposition}\label{P2}
Let $Y \subseteq [0,1]$ contain as a subset  $\left\{ \frac{1}{2}+\frac{1}{k+1}: k \in \N \right\} \cup \left\{ \frac{1}{2}-\frac{1}{k+1}: k \in \N \right\}$. 
Then any SWO defined on  $X=Y^{\N}$ satisfying fixed step asymptotic Pigou-Dalton (s-APD) and anonymity (AN) axioms is not representable.
\end{proposition}

\section{Conclusions}\label{s5}
In this paper, we have proposed a new version of Gini Coefficient. 
This index represents social welfare orders satisfying generalized Pigou-Dalton transfer principle and anonymity on the space of infinite utility streams when individual agents' utility is assigned values from a finite set $Y\subset \R$.
Since an explicit formula for the index is described, it is useful for policy formulation.
We also show that when we consider more general set $Y$ (i.e., $Y$ having infinitely many elements of the type considered in Proposition \ref{P2}) real-valued representation is impossible. 
It is an open question for us to explore in future if social welfare function exists in case $Y(<)$ is a well-ordered infinite subset of real numbers.

\bibliographystyle{plainnat}
\setlength{\bibsep}{0pt}
\small{\bibliography{APAnonymity}}

\end{document}